\documentclass{scrartcl}

\usepackage[hidelinks]{hyperref}
\usepackage{url}
\usepackage{authblk} 

\usepackage{amsmath}
\usepackage{amsthm}
\newtheorem{lemma}{Lemma}
\newtheorem{theorem}{Theorem}

\usepackage{amssymb}
\def\N{{\mathbb{N}}}
\def\proba{{\mathbb{P}}}
\def\E{{\mathbb{E}}}
\def\A{\mathbf{A}}
\def\B{\mathbf{B}}
\def\service{\mathbf{S}}
\def\X{\mathbf{X}}
\def\Y{\mathbf{Y}}

\usepackage[xindy]{glossaries}	
\newacronym{SNC}{SNC}{Stochastic Network Calculus}

\usepackage{graphics,graphicx}
\usepackage{subfig}
\graphicspath{{../figures/}}
\DeclareGraphicsExtensions{png.,.pdf,.jpeg}

\usepackage{xcolor}
\xdefinecolor{myred}{RGB}{235,62,62}
\xdefinecolor{myblue}{RGB}{92,92,235}
\xdefinecolor{myorange}{RGB}{245,138,42}

\usepackage{tikz}
\usetikzlibrary{calc}
\usetikzlibrary{shapes.multipart} 
\usetikzlibrary{intersections}

\newcommand\fcfs[1][]{
  \begin{tikzpicture}
    \node[
      draw,
      inner sep=.12cm,
      rectangle split, rectangle split parts=#1, rectangle split horizontal,
      rectangle split empty part width=-.1cm, rectangle split empty part height=.7cm,
    ] (queue) {};
    \fill[white] ($(queue.north west)+(-.05cm,.05cm)$)
    rectangle ($(queue.south west)+(.05cm,-.05cm)$);
    \useasboundingbox ($(queue.south west)+(.05cm,0)$)
    rectangle (queue.north east);
  \end{tikzpicture}
}

\tikzset{
  server/.style = {draw, circle, minimum size=1cm, anchor=west},
}

\usepackage{pgfplotstable}
\usepackage{pgfplots}
\pgfplotsset{tikzDefaults/.style={
		enlargelimits=0,
		ymin = .001,
		ymax = 1,
		xmin = 0,
		xmax = 80,
		mark size=2pt,
		line width = 1,
		xticklabel style = {font=\footnotesize, yshift=.1cm},
		yticklabel style = {font=\footnotesize, xshift=.1cm},
		label style = {font=\footnotesize},
		legend style = {font=\footnotesize},
	},
	tikzLocality/.style={
		enlargelimits=0,
		ymin = .0001,
		ymax = 1,
		xmin = 0,
		xmax = 80,
		xticklabel style = {font=\footnotesize, yshift=.1cm},
		yticklabel style = {font=\footnotesize, xshift=.1cm},
		label style = {font=\footnotesize},
		legend style = {font=\footnotesize},
		legend pos = north east,
		line width = 1,
		x label style={at={(axis description cs:0.5,-0.1)}},
		y label style={at={(axis description cs:-0.15,.5)}},
	},
	doobs/.style = {myorange, densely dashed},
	noyau/.style = {myblue},
	simu/.style = {green!70!black, densely dotted},
	line_c/.style = {myblue, densely dashdotted},
}

\pgfplotsset{compat=newest}

\begin{document}

\title{Of Kernels and Queues: when network calculus meets analytic combinatorics}

\author[1]{Anne Bouillard}
\author[1,2]{C\'eline Comte}
\author[1]{\'Elie de Panafieu}
\author[1]{Fabien~Mathieu\thanks{The authors are members of LINCS, see \href{http://www.lincs.fr}{http://www.lincs.fr}.}}

\affil[1]{Nokia Bell Labs, France}
\affil[2]{T\'el\'ecom ParisTech, France}

\subtitle{NetCal 2018 (Author version)}

\maketitle

\begin{abstract}
	Stochastic network calculus is a tool for computing error bounds on the performance of queueing systems.
  However, deriving accurate bounds for networks consisting of several queues or subject to non-independent traffic inputs is challenging.
  In this paper,
  we investigate the relevance
  of the tools from analytic combinatorics,
  especially the \emph{kernel method},
  to tackle this problem.
  Applying the kernel method
  allows us to compute the generating functions of the queue state distributions
  in the stationary regime of the network.
  As a consequence,
  error bounds with an arbitrary precision can be computed.
  In this preliminary work,
  we focus on simple examples which are representative of the difficulties that the kernel method allows us to overcome. 
\end{abstract}

\section{Introduction}
\label{eq:intro}

The development of new wireless communication technologies (5G) shed a new light on queueing theory, as the strong requirements on buffer occupancy, latencies, and reliability,
bring the need for accurate dimensioning rules. In many scenarios, data packets arrive by batches and are processed by a server that can deal with a fixed number of packets per time slot~\cite{Sehier17}. The $ G/D/1 $ queue is thus a natural model.
	
A powerful tool to analyze such queues is \gls{SNC}~\cite{FR15}.
The aim of \gls{SNC} is to derive precise error bounds on the performance of systems,
combining deterministic network calculus and probabilistic tools.

Among the techniques developed so far, the Tailbounded approach~\cite{JY08} introduces a violation probability in the parameters of the deterministic setting.
It makes possible the computation of error bounds in networks, like in~\cite{CBL06}, but these bounds are usually loose.
A second technique, introduced in~\cite{Chang2000}, relies on moment generating functions. It can be very accurate for one queue.
For example, in~\cite{PC14, CP15}, the authors obtain tight upper and lower bounds
for the single-server case under various service policies and arrival processes,
using martingales and Doob's inequality.
However, for more general topologies, the method becomes non applicable due to interdependencies between the processes.
Recently, some (looser) bounds have been computed using H\"older's inequality~\cite{Beck16, NS17}.

The use of generating functions to investigate random processes
is the core principle of \emph{analytic combinatorics},
a subfield of combinatorics (see \cite{FS09}).
This community developed mathematical tools to study random walks \cite{FG00},
such as the \emph{kernel method} \cite{BaFl02, BoMi10}, described later.
The link between random walks and queueing theory is known
and results on the former were transferred to the latter \cite{FaIaMa17}.

In this article, we show how generating functions
and the kernel method can be applied to derive precise results on queueing systems.

In \S\ref{sec:gf}, we first recall the main definitions and notations of generating functions.
The main contribution of the paper is given in \S\ref{sec:oneone}, where we show in detail how to apply the kernel method to study the $GI/D/1$ queue.
Although the result itself is well-known (we retrieve the Pollaczek-Khinchine formula), the interest of the analysis is that it contains all the pieces for further extensions, such as several flows of packets, several queues, or non i.i.d.\ arrivals. Some of these extensions are developed in \S\ref{sec:ext}: random service, multi-flow and multi-queue.
Finally, we confront our results with simulations in \S\ref{sec:simu}.

\section{Generating functions}
\label{sec:gf}

In this section, we recall some basics of generating functions. Let $(a_n)_{n \geq 0}$ be a sequence of non-negative numbers. Its generating function is the formal series 
$$A(u) = \sum_{n \geq 0} a_n u^n.$$
The $n$-th monomial $a_n$ will also be denoted by $[u^n]A(u)$.
In combinatorics, $a_n$ is often the number of objects of size $n$ within a given family.
In probability, $a_n$ is usually the probability that a random variable $\A$ with values in $\N$ is equal to $n$:
$$A(u) = \sum_{n \geq 0} \proba(\A = n) u^n.$$
In that case, we write $\A \sim A$;
the convergence radius $\rho$ of the function $A$ is at least $1$, $A(1) = 1$, $A'(1) = \E[\A]$,
and we assume $\lim_{u \to \rho}A(u) = +\infty$ to simplify the asymptotic analyses.

Two elementary operations can be performed on generating functions. Suppose that $A$ and $B$ are the generating functions of two random variables $\A$ and $\B$, respectively. 
\begin{enumerate}
	\item If the events $\{\A=n\}$ and $\{\B=n\}$ are disjoint for each $n \in \N$, then \\
    $$
    A(u) + B(u) = \sum_{n \geq 0} \proba(\{\A=n\} \cup \{\B=n\})u^n.
    $$
	\item If $\A$ and $\B$ are independent, then $A(u) B(u)$ is the generating function of the r.v. $\A+\B$.
\end{enumerate}

Consider the example of a Galton-Watson tree,
which is a branching process where the number of children of each node is i.i.d. with distribution given by the generating function $A$. 
The number of nodes of the tree is
$$\X = 1 + \sum_{k=1}^{\A} \X_k,$$
where $\A \sim A$ is the number of children of the root
and $\X_k$ is the number of nodes in the subtree rooted at the $k$-th child of the root.
$\X_k$ has the same distribution as $\X$, hence the same generating function, denoted by $T_A$. Therefore, we obtain 
\begin{equation}
\label{eq:gw}
T_A(u) = u A(T_A(u)).
\end{equation}
This equation characterizes $T_A$.
\figurename~\ref{fig:Tu} shows how $T_A(u)$ is computed.
Being the generating function of a probability distribution,
$T_A(u)$ must be a solution of Eq.~\eqref{eq:gw} that is analytic at $0$,
also known as a \emph{small root} of the equation.
$T_A(u)$ is then the abscissa coordinate of the first intersection of $A(x)$ with the line $x/u$. There is a maximal value $\rho_{T_A}>1$ of $u$ for which a root exists.

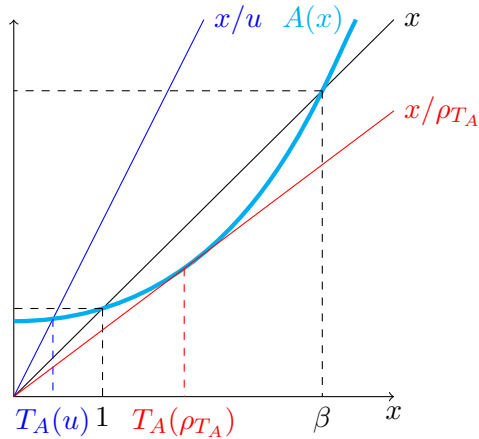
\begin{figure}[htbp]
	\centering
	\begin{tikzpicture}	
	\draw[->, name path=axex] (0,0) -- (5,0)node[below] {$x$};
	\draw[->, name path=axey] (0,0) -- (0,5);
	\draw[name path=id] (0,0) -- (5,5) node[right] {$x$};
	\draw[blue, name path=id2] (0,0) -- (2.5,5) node[right] {$x/u$};
	\draw[ultra thick,cyan, name path=phi] (0,1) .. controls (3,1)  and (4,4) .. (4.5,5) node[left] {$A(x)$};
	\path[name intersections={of=id and phi}];
	\coordinate (Un)  at (intersection-1);
	\coordinate (beta)  at (intersection-2);
	\path[name intersections={of=id2 and phi}];
	\coordinate (Uz)  at (intersection-1);
	
	\path[name path=h1] (beta) --++ (-4.3,0);
	\path[name intersections={of=h1 and axey}];
	\coordinate (ybeta) at (intersection-1);
	\draw[dashed] (beta) -- (ybeta);
	
	\path[name path=v1] (beta) --++ (0,-4.3);
	\path[name intersections={of=v1 and axex}];
	\coordinate (xbeta) at (intersection-1);
	\draw[dashed] (beta) -- (xbeta) node[below] {$\beta$};
	
	\path[name path=h2] (Un) --++ (-1.3,0);
	\path[name intersections={of=h2 and axey}];
	\coordinate (yUn) at (intersection-1);
	\draw[dashed] (Un) -- (yUn);
	
	\path[name path=v2] (Un) --++ (0,-1.3);
	\path[name intersections={of=v2 and axex}];
	\coordinate (xUn) at (intersection-1);
	\draw[dashed] (Un) -- (xUn) node[below] {$1$};
	
	\path[name path=v3] (Uz) --++ (0,-1.2);
	\path[name intersections={of=v3 and axex}];
	\coordinate (xUz) at (intersection-1);
	\draw[dashed, blue] (Uz) -- (xUz) node[below] {$T_A(u)$};
	
	\draw[red, name path=id0] (0,0) -- (5,3.79) node[right] {$x/\rho_{T_A}$};
	\path[name intersections={of=id0 and phi}];
	\coordinate (U0)  at (intersection-1);
	\path[name path=v4] (U0) --++ (0,-1.8);
	\path[name intersections={of=v4 and axex}];
	\coordinate (xU0) at (intersection-1);
	\draw[dashed, red] (U0) -- (xU0) node[below] {$T_A(\rho_{T_A})$};
	
	\end{tikzpicture}
	\caption{Equation $T_A(u) = uA(T_A(u))$.}
	\label{fig:Tu}
\end{figure}

By deriving both sides at $u=1$,
one gets that $\E[\X] = 1 + \E[\X]\E[\A]$,
so that $\E[\X] = (1-\E[\A])^{-1}$ if $\E[\A] <1$.

Adopting a combinatorics viewpoint allows us to consider generating functions that do not represent a distribution. For example, $T_A^k$ is the generating function of the distribution of the total size of $k$ independent Galton-Watson trees, and $\frac{1}{1-T_A} = \sum_{k\geq 0} T_A^k$ is the sum of the distributions for all possible $k$. Let $(\X_j)_{j \geq 0}$ denote the sequence of the sizes of i.i.d.~Galton-Watson trees, then
$$[u^n]\frac{1}{1-T_A(u)} = \proba(\exists k,~ \X_1 + \cdots + \X_k = n).$$
Studying the behavior of this series will prove useful to derive the asymptotic probability that an arbitrary number of trees has a given total size. As $T_A(u)<1$ for all $0\leq u<1$ and $T_A(1)=1$, we can apply the result of~\cite[p. 294, Th. V.1]{FS09}: 
\begin{equation}
\label{eq:asympt-trees}
[u^n]\frac{1}{1-T_A(u)} \underset{n\to \infty}{\sim} \frac{1}{T_A'(1)} = 1-\E[\A].
\end{equation}

All the definitions can be extended to the multivariate case.

\section{The single-server queue}
\label{sec:oneone}

In this section, we present the simple example of a single-server
queue with one flow of packets, as depicted in \figurename~\ref{fig:oneone}. The results presented here are not new (we eventually rediscover the Pollaczek-Khinchine formula, and apply tools developed by \cite{BaFl02}), but our aim is to present the method that will be generalized later.

\subsection{Queueing model}
\label{subsec:model}

The queue is initially empty.
At each time slot $t\geq 1$, one packet (if any) is served and then $\A_t$ packets arrive.
The sequence $(\A_t)_{t\geq 1}$ is i.i.d.\ with generating function $A$ and mean $\lambda < 1$.
We let $\X_t$ denote the number of packets in the queue at the end of time slot $t$.
The system is driven by the equations
\begin{equation}
  \label{eq:dyn1}
  \X_0 = 0 \text{ and }\X_{t+1} = (\X_t - 1)_+ + \A_{t+1},
  \quad \forall t \geq 0,
\end{equation}
where $(\cdot)_+ = \max(\cdot,0)$.

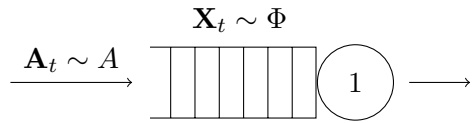
\begin{figure}[htbp]
  \centering
  \begin{tikzpicture}
    \node[inner sep=0] (queue) {\fcfs[7]};
    \node[anchor=south] at ($(queue.north)+(.2cm,0)$) {$\X_t \sim \Phi$};

    \draw[->] ($(queue.west)-(1.6cm,0)$) --
    node[midway, above] {$\A_t \sim A$}
    (queue.west);

    \node[server] (server) at (queue.east) {1};
    \draw[->] ($(server.east)+(.2cm,0)$) -- ($(server.east)+(1cm,0)$);
  \end{tikzpicture}
  \caption{A single-server queue crossed by a single flow.
  The server processes one packet during each time slot.}
  \label{fig:oneone}
\end{figure}

\subsection{Generating function}
\label{subsec:gen}

We define the generating function of the queue state as
\begin{equation}
  \label{eq:oneone-phidef}
  \Phi(u,z) = \sum_{n \geq 0} \sum_{t \geq 0} \proba(\X_t = n) u^n z^t.
\end{equation}
For each $t \in \N$,
taking the coefficient of $z^t$ yields
$$
\left[ z^t \right]\Phi(u,z) = \sum_{n \geq 0} \proba(\X_t = n) u^n,
$$
which is the generating function of $\X_t$.
We will show that

\begin{lemma}
	\begin{equation}
	\label{eq:oneone-intuitive}
	\Phi(u,z) =
	1 + zA(u)
	\left[ (\Phi(u,z) - \Phi(0,z))u^{-1} 
	+ \Phi(0,z) \right].
	\end{equation}
\end{lemma}

\begin{proof}[Sketch of proof]
Eq.~\eqref{eq:dyn1} implies $\X_0 = 0$ and for any $t \geq 1$, with the notation $a_n = \proba(\A = n)$, 
\begin{align*}
	\proba(\X_t = n-1) =\ & a_{n-1} (\proba(\X_{t-1} = 0) + \proba(\X_{t-1} = 1))
    + \sum_{m=2}^n a_{n-m} \proba(\X_{t-1} = m).
\end{align*}
Multiplying this relation by $u^n z^t$, summing over $n$, $t$ and dividing both sides by $u$ leads to the equation of the lemma.
The formula can also be directly derived using the \emph{Symbolic Method} \cite{FS09}.
\end{proof}

Eq.~\eqref{eq:oneone-intuitive} can be rewritten as
\begin{equation}
\label{eq:oneone-practical}
\Phi(u,z) \left[ 1 - z A(u)u^{-1} \right]
  = 1 - \Phi(0,z) z A(u) \left[ u^{-1} - 1 \right].
\end{equation}

Although Eq.~\eqref{eq:oneone-practical}
completely characterizes $\Phi(u,z)$,
it is not straightforward
to derive an explicit formula for $\Phi(u,z)$ from it, as we would need an expression for $\Phi(0,z)$. This expression will be obtained with the kernel method. 

\subsection{Kernel method}
\label{subsec:kernel}

When the left-hand side of Eq.~\eqref{eq:oneone-practical} cancels, so does the right-hand side. The kernel method \cite{BaFl02, BoMi10} consists in taking $u = U(z)$ such that the second factor of the left-hand side cancels. Here, $U(z)$ is implicitly defined by the equality
$U(z) = z A(U(z))$, and we recognize from Eq.~\eqref{eq:gw} the size distribution
of a Galton-Watson tree where the offspring distribution has the generating function $A$.
Therefore, we have $U = T_A$.

Injecting $T_A(z)$ in Eq.~\eqref{eq:oneone-practical}
cancels its left-hand side, and its right-hand side can be rewritten as
\begin{equation*}
  \Phi(0,z) = \frac1{1 - T_A(z)}.
\end{equation*}
Going back to Eq.~\eqref{eq:oneone-practical}, we finally obtain
\begin{equation}
  \label{eq:oneone-phires}
  \Phi(u,z) = 
  \frac
  { 1 + \frac1{1 - T_A(z)} z A(u) \left( 1 - u^{-1} \right) }
  { 1 - z A(u)u^{-1} }.
\end{equation}

The kernel method has the following interpretation in terms of the queue sample paths.
The generating function $\Phi(0,z) = \sum_{t \geq 0} \proba(\X_t = 0) z^t$ is associated with the probability of having an empty queue.
Consider the duration between two consecutive instants when the queue is empty, which we call an \emph{inter-empty period}.
It was showed in~\cite{Kendall} that
we can build a Galton-Watson tree with offspring distribution $A$ from an inter-empty period:
each node represents a time slot;
its children are the time slots when the packets arrived during this time slot are served.
Having an empty queue at time $t$ means that
the realization between times $0$ and $t$ is made up of an arbitrary number of inter-empty periods.
This corresponds exactly to $\frac{1}{1-T_A(z)}$, where $T_A$ is as defined in Eq.~\eqref{eq:gw}.

\subsection{Asymptotic performance}
\label{subsec:perf}

In this paragraph, our aim is to bound the probability that $\X_t$ exceeds some value $R$ in stationary regime.
Note that, by monotony, this will also be an upper bound for the initially empty queue.
We proceed in two steps.
We first compute $\Pi$, the generating function of the stationary distribution of $(\X_t)$,
and then we derive the asymptotic behavior of $\Pi$.

\paragraph{Computation of $\Pi$}

We know that, under the stability condition $A'(1) = \lambda < 1$, the distribution of $\X_t$ converges to a stationary distribution $\pi$ as $t$ tends to $+\infty$.
The first step of our analysis consists in finding the generating function $\Pi$ of this distribution $\pi$.
Recall that, for each $t \in \N$, the generating function of $\X_t$ is $\Pi_t(u) = [z^t]\Phi(u,z).$
By~\cite[p. 624]{FS09}, it suffices to study
the limit of $\Pi_t(u)$ as $t$ tends to $+\infty$, when $u$ is fixed.
The obtained limit is exactly $\Pi(u)$.

Let us fix $u = u_0$.
We see in Eq.~\eqref{eq:oneone-phires} that $\Phi(u_0,z)$ has two potential poles, 1 and $u_0/A(u_0)$.
It can be checked that $T_A(\frac{u_0}{A(u_0)}) = u_0$, so that
$u_0 / A(u_0)$ is actually not a pole.
In order to derive the asymptotic behavior of $\Pi_t(u_0)$ as $t$ tends to $+\infty$,
we first compute a simpler equivalent of $\Phi(u_0,z)$ in the neighborhood of its pole $z=1$.
After some rewriting, we obtain
$$\Phi(u_0,z) = \frac{u_0}{u_0-zA(u_0)} + \frac1{1-T_A(z)} \frac{zA(u_0)(u_0-1)}{u_0-zA(u_0)}.$$
As a consequence, 
$$\Phi(u_0,z) \underset{z\to 1}{\sim}  \frac{u_0}{u_0-A(u_0)} + \frac{1}{1-T_A(z)} \frac{A(u_0)(u_0-1)}{u_0-A(u_0)},$$
and from Eq.~\eqref{eq:asympt-trees}, the terms are equivalent to 
$$[z^t]\Phi(u_0,z) \underset{t\to \infty}{\sim} (1-\lambda) \frac{A(u_0)(u_0-1)}{u_0-A(u_0)}.$$

Therefore, the generating function of $\pi$ is equal to the one given by the Pollaczek-Khinchine formula 
$$\Pi(u) = (1-\lambda) \frac{A(u)(u-1)}{u-A(u)}.$$

\paragraph{Performance} 

The second solution $\beta$ of the equation $u = A(u)$ is the convergence radius of the function $\Pi$
(with $\beta = +\infty$ in the degenerate case where $A(u)$ is linear).
The error bound, \textit{i.e.}~the probability that the buffer occupancy is at least $R$, is $\sum_{n\geq R} \pi(n)$. Its generating function is
\begin{align}
\label{eq:error}	
 E(u) = \sum_{R\geq 0} \bigg(\sum_{n\geq R} \pi(n)\bigg)u^R & = \frac{1-u\Pi(u)}{1-u}.
\end{align} 
The asymptotic analysis of this generating function yields
\begin{theorem}
With $\X\sim \Pi$,
\begin{equation}
\label{eq:kernel_simple}
  \proba(\X \geq R) \underset{R\to\infty}{\sim} (1-\lambda)\frac{\beta}{A'(\beta)-1} \beta^{-R} \text{.}
\end{equation}
\end{theorem}

\section{Extensions of the single-server queue}
\label{sec:ext}

The analysis in the previous section shows that
deriving an equation satisfied by the generating function from the system dynamics is the easy step;
solving this equation is harder.
We now consider a few simple extensions of the model of \S \ref{subsec:model},
where the kernel method allows us to perform the analysis and derive
explicit formulas for the performance metrics.

\subsection{Random service}
\label{subsec:oneone-p}

We consider a first extension of the model of \S \ref{subsec:model}
where the service is random.
Specifically,
at each time slot $t \ge 1$,
the server processes one packet (if any)
with some probability $p > \lambda$,
and zero packet otherwise.
The system is driven by the equations
$$
\X_0 = 0
~\text{and}~
\X_{t+1} = (\X_t - \service_t)_+ + \A_{t+1},
\quad \forall t \geq 0,
$$
where $(\service_t)_{t \in \N}$ is a sequence of independent,
Bernoulli distributed random variables with parameter $p$.
The corresponding generating function is $S(u) = 1 - p + pu$.

The generating function $\Phi$ of the system state is
as defined in \eqref{eq:oneone-phidef}.
The equation satisfied by $\Phi$
is a rewriting of Eq.~\eqref{eq:oneone-practical}, where $u^{-1}$ is replaced by $S(u^{-1})$:
\begin{align*}
  \Phi(u,z) [ 1 - z A(u) S( u^{-1} ) ]
  = 1 - \Phi(0,z) z A(u) [ S( u^{-1} ) - 1 ].
\end{align*}

Applying the kernel method consists in choosing $u = U(z)$
such that 
$z A(U(z)) S(U(z)^{-1}) = 1$.
We can rewrite this as
$U(z) = G(zA(U(z)))$,
where $G$ is the generating series of the geometric distribution (defined on the set of positive integers) with parameter $p$:
$$
G(s) = \frac{ps}{1 - (1-p)s}.
$$
In much the same way as in \S \ref{subsec:kernel}, we obtain
$$
\Phi(0,z)
= \frac1{1 - z A(U(z))}
= \frac1{1 - T_{A \circ G}(z)}.
$$
The second equality
holds because $\bar{U}(z) = z A(U(z))$
satisfies the equation $\bar{U}(z) = z A(G(\bar{U}(z)))$,
so that $\bar{U}$ is also the generating function
of the size of a Galton-Watson tree
with offspring distribution $A \circ G$.
Finally, we obtain
$$
\Phi(u,z) = \frac
{ 1 + \frac1{1 - T_{A \circ G}(z)} z A(u) (1 - S(u^{-1})) }
{ 1 - z A(u) S(u^{-1}) }.
$$
The interpretation is similar to that of \S \ref{subsec:kernel},
except that the number of time slots dedicated to a given customer is now
geometrically distributed with parameter $p$.

The stationary distribution has the generating function
$$\Pi(u) = \bigg( 1-\frac{\lambda}p \bigg) \frac{A(u)\left( 1-S(u^{-1}) \right)}{1-A(u)S(u^{-1})}.$$
Let $\gamma$ be the largest solution of the equation $A(u) S(u^{-1}) = 1$,
or, equivalently, $u = G(A(u))$.
Similarly to \S \ref{subsec:perf},
we obtain
\begin{theorem}
	For $\X\sim \Pi$,
$$
\proba(\X\ge R)
\underset{R\to \infty}{\sim}
\frac{ (1-\frac\lambda{p}) (A(\gamma)-1) }
{ (\gamma - 1) \big( A'(\gamma) S(\gamma^{-1}) - A(\gamma)p\gamma^{-2} \big)}
\gamma^{-R}.
$$
\end{theorem}

The kernel method is also applicable to the case where
the server processes up to $c$ packets at each time slot, for some integer $c \geq 1$.
Assume that the distribution of the number of served packets
has generating function $S(u)$ if the queue contains at least $c$ packets,
and $S_k(u)$ if the queue contains exactly $k$ packets, for each $0 \le k < c$.
The equivalent of Eq.~\eqref{eq:oneone-practical} is now
\begin{align}
  \label{eq:oneone-multiservice}
  \Phi(u,z) &[ 1 - z A(u) S( u^{-1} ) ] =
  1 - \sum_{k=0}^{c-1} z A(u) ( S( u^{-1} ) - S_k( u^{-1} ) )
  \frac{u^k}{k!} \frac{\partial^k}{(\partial u)^k} \Phi(0,z).
\end{align}
We refer the reader to \cite[p.~508]{FS09} or \cite{BaFl02} for a detailed analysis of this equation,
and provide here a short version.
There are $c$ independent functions $(U_k(z))_{0 \leq k < c}$, analytic at $0$,
that cancel the second term of the left hand-side, because $S$ is a degree $c$ polynomial.
Thus, we obtain $c$ equations for the $c$ unknowns $\frac{\partial^k}{(\partial u)^k} \Phi(0,z)$ for $0 \leq k < c$.
Solving this system of equations and injecting the solution in Eq.~\eqref{eq:oneone-multiservice}
leads to the expression of $\Phi(u,z)$.

\subsection{Several flows with priorities}
\label{subsec:onetwo}

Consider the system in \figurename~\ref{fig:2flows}.
As in \S \ref{subsec:model},
the queue is initially empty and the server processes one packet at each time slot.
There are two flows of packets.
Flow $1$ has priority over flow $2$,
so that a packet from flow $1$ is served
whenever the queue contains at least one packet from this flow at the beginning of this time slot.
At each time slot $t\geq 1$, $\A_t$ packets from flow 1 and $\B_t$ packets from flow 2 arrive.
The sequences $(\A_t)_{t \ge 1}$ and $(\B_t)_{t \ge 1}$ are independent and i.i.d.\ with generating function $A$ and $B$ and mean $\lambda_A$ and $\lambda_B$, respectively,
such that $\lambda_A + \lambda_B < 1$.
We denote by $\X_t$ and $\Y_t$ the respective numbers of packets from flows $1$ and $2$ in the queue at the end of time slot $t$.
The system is then driven by the equations $\X_0 = 0$, $\Y_0 = 0$ and
\begin{equation}
	\label{eq:dyn2}
	\begin{cases}
    \X_{t+1} = (\X_t - 1)_+ + \A_{t+1}, \\
    \Y_{t+1} = (\Y_t - 1_{\{\X_t = 0\}})_+ + \B_{t+1},
	\quad \forall t \ge 0.
	\end{cases}
\end{equation}

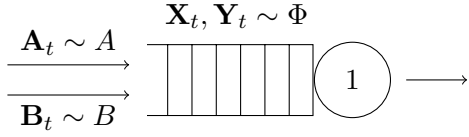
\begin{figure}
	\centering
	\begin{tikzpicture}
	\node[inner sep=0] (queue) {\fcfs[7]};
	\node[anchor=south] at ($(queue.north)+(.2cm,0)$) {$\X_t, \Y_t \sim \Phi$};
	
	\draw[->] ($(queue.west)-(1.6cm,0)+(0,.2cm)$) --
	node[midway, above] {$\A_t \sim A$}
	($(queue.west)+(0,.2cm)$);
	\draw[->] ($(queue.west)-(1.6cm,0)-(0,.2cm)$) --
	node[midway, below] {$\B_t \sim B$}
	($(queue.west)-(0,.2cm)$);
	
	\node[server] (server) at (queue.east) {1};
	\draw[->] ($(server.east)+(.2cm,0)$) -- ($(server.east)+(1cm,0)$);
	\end{tikzpicture}
  \caption{A single-server queue crossed by two flows.}
	\label{fig:2flows}
\end{figure}

We define a generating function for the state $(\X_t,\Y_t)_{t \ge 0}$,
with three variables $u$, $v$, and $z$, respectively representing the numbers of packets from flows 1 and 2 and the time:
$$
\Phi(u,v,z) = \sum_{n \ge 0} \sum_{m \ge 0} \sum_{t \ge 0}
\proba(\X_t = n, \Y_t = m) u^n v^m z^t.
$$

The equation satisfied by $\Phi$ follows from Eq.~\eqref{eq:dyn2}:
\begin{align*}
	\Phi(u,v,z) = 1 + z A(u)B(v)\big[
  &(\Phi(u,v,z) - \Phi(0,v,z))u^{-1}   \\
  &+  (\Phi(0,v,z) - \Phi(0,0,z))v^{-1}  
	+  \Phi(0,0,z)\big].
\end{align*} 
Here $(\Phi(u,v,z) - \Phi(0,v,z))u^{-1}$ represents the service of a packet from flow 1 (if any)
and $(\Phi(0,v,z) - \Phi(0,0,z))v^{-1}$ the service of a packet from flow 2 (if any, and when there is no packet from flow $1$).
This equation can be rewritten as
\begin{multline}
  \Phi(u,v,z) \left[ 1 - z A(u) B(v) u^{-1} \right] = 1 - z A(u) B(v) \times \\ 
  \label{eq:onetwo-practical}
  \left[
  \Phi(0,v,z) \left( u^{-1} - v^{-1} \right) 
  + \Phi(0,0,z) \left( v^{-1} - 1 \right)
  \right].
\end{multline}

We could find an expression for $\Phi(0,0,z)$ by applying twice the kernel method on this equation,
but we prefer a more intuitive approach.
The generating function $\Phi(0,0,z)$ is associated with the probability that the queue is empty.
This probability only depends on the global arrival process $(\A_t + \B_t)_{t \ge 1}$,
regardless of the division of packets into flows.
Therefore, we know from \S \ref{subsec:kernel} that
$$\Phi(0,0,z) = \frac{1}{1-T_{AB}(z)},$$
where $T_{AB}$ is the generating function of the size of a Galton-Watson tree with offspring distribution $AB$. 

We apply the kernel method to derive the expression of $\Phi(0,v,z)$.
Let us take $u = U(v,z)$ such that $U(v,z) = z B(v) A(U(v,z))$,
in order to cancel the left-hand side of Eq.~\eqref{eq:onetwo-practical}.
We obtain $U(v,z) = T_A(z B(v))$, which is again strongly related to a Galton-Watson tree.
The interpretation is similar to that of \S \ref{subsec:kernel},
except that $U(v,z)$ is now the generating function of the number of time slots passed and flow-$2$ packets arrived
during an inter-empty period of flow $1$.
The priority of flow $1$ ensures that no packet from flow $2$ is served in the meantime.
After simplifications, we get 
$$
\Phi(0,v,z) = \frac
{v +  \frac1{1-T_{AB}(z)} T_A(z B(v)) (v-1)}
{v -  T_A(z B(v))},
$$
and the expression for $\Phi(u,v,z)$ immediately follows.
It is not difficult to see that this method can be generalized to queues with more than two flows with a total order on priority levels. 

Suppose that we focus on the number of packets from flow $2$ in the stationary state.
We are then interested in the generating function $\Phi(1,v,z)$.
The same approach as in \S\ref{subsec:perf} can be used to obtain the following result.
\begin{theorem}
	The limit distribution of the number of packets of flow $2$ is given by the generating series
	$$\Pi(v) = (1-\lambda_A-\lambda_B) \frac
	{B(v) (1-v) (T_A(B(v)) -1 )}
	{(1 - B(v))(v - T_A(B(v))}.
  $$
\end{theorem}

Let $\delta$ be the largest solution of the equation $v = T_A(B(v))$
(as $T_A$ and $B$ are convex, there are exactly 2 solutions, and the smallest is 1).
Similarly to \S \ref{subsec:perf},
\begin{theorem}
	For $\Y\sim \Pi$,
\begin{equation}
  \label{eq:double_file}
  \proba(\Y\geq R) \underset{R\to\infty}{\sim} \frac{ (1-\lambda_A-\lambda_B)B(\delta)(\delta-1)}{(1-B(\delta))  (1- (T_A\circ B)'(\delta))} \delta^{-R} \text{.}
\end{equation}
\end{theorem}

\subsection{Several queues}
\label{subsec:several}

We can also use generating functions to describe the dynamics of networks of queues.
As an example, consider the network of \figurename~\ref{fig:2queues},
consisting of two single-server queues.
At each time slot $t \ge 1$,
$\A_t$ packets arrive at queue $1$
and $\B_t$ packets arrive at queue $2$.
As before,
the sequences $(\A_t)_{t \ge 1}$ and $(\B_t)_{t \ge 1}$
are independent and i.i.d.\ with generating function $A$ and $B$
and mean $\lambda_A$ and $\lambda_B$,
respectively,
such that $\lambda_A + \lambda_B < 1$.
Additionally, the packets served at queue $1$ are subsequently forwarded to queue $2$ for service.
We let $\X_t$ and $\Y_t$ denote the numbers of packets at queues $1$ and $2$, respectively, at time $t$.

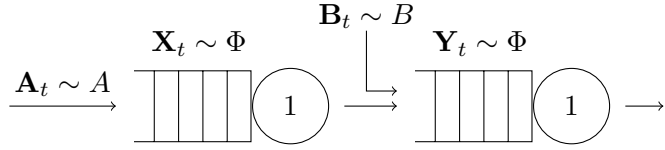
\begin{figure}[htbp]
  \centering
  \begin{tikzpicture}
    \node[inner sep=0] (queue1) {\fcfs[5]};
    \node[anchor=south] at ($(queue1.north)+(.2cm,0)$) {$\X_t \sim \Phi$};
    \node[server] (server1) at (queue1.east) {1};

    \node[inner sep=0] (queue2) at ($(queue1.east)+(2.8cm,0)$) {\fcfs[5]};
    \node[anchor=south] at ($(queue2.north)+(.2cm,0)$) {$\Y_t \sim \Phi$};
    \node[server] (server2) at (queue2.east) {1};

    \draw[->] ($(queue1.west)-(1.4cm,0)$) --
    node[midway, above] {$\A_t \sim A$}
    (queue1.west);
    \draw[->] ($(queue2.west)-(.4cm,0)+(0,1cm)$) |-
    node[near start, above, yshift=.3cm] {$\B_t \sim B$}
    ($(queue2.west)+(0,.2cm)$);

    \draw[->] ($(server1.east)+(.2cm,0)$) -- (queue2.west);

    \draw[->] ($(server2.east)+(.2cm,0)$) -- ($(server2.east)+(.7cm,0)$);
  \end{tikzpicture}
  \caption{A network of two queues crossed by two flows.}
  \label{fig:2queues}
\end{figure}

The dynamics of the system, which is initially empty, are driven by the equations $\X_0 = 0$, $\Y_0 = 0$ and
\begin{equation}
\label{eq:dyn2q}
  \left\{
    \begin{array}{ll}
      \X_{t+1} & = (\X_t -1)_+ + \A_{t+1}, \\
      \Y_{t+1} & = (\Y_t - 1)_+ + \B_{t+1} + 1_{\{\X_t>0\}},
      \quad \forall t \ge 0.
    \end{array}
  \right.
\end{equation}
We define a generating function for the state $(\X_t,\Y_t)_{t \ge 0}$,
with three variables $u$, $v$, and $z$, respectively representing the numbers of packets at queues $1$ and $2$ and the time:
$$
\Phi(u,v,z) = \sum_{n \ge 0} \sum_{m \ge 0} \sum_{t \ge 0}
\proba(\X_t = n, \Y_t = m) u^n v^m z^t.
$$
This function $\Phi$ satisfies the following equation:
\begin{align}
  \nonumber
  \Phi(u,v,z) [ 1 - z A(u) B(v)u^{-1} ]
  = 1 - &z A(u) B(v) \\
  \nonumber
  &[
  \Phi(u,0,z) ( u^{-1} - vu^{-1} )
  + \Phi(0,v,z) ( u^{-1} - v^{-1} ) \\
  \label{eq:2-2}
  &\phantom{[}+ \Phi(0,0,z) ( vu^{-1} + v^{-1} - u^{-1} - 1 ) ].
\end{align}

The kernel method cannot be applied directly.
Indeed, we need to compute three generating functions
($\Phi(u,0,z)$, $\Phi(0,v,z)$, and $\Phi(0,0,z)$),
while we can only apply the kernel method (at most) twice.
It is, however, possible to find an additional relation between these functions: 
\begin{align*}
  \Phi(u,0,z) &= 1 + zA(u)B(0) \big[ \Phi(0,0,z) + [v^1]\Phi(0,v,z) \big].
\end{align*}
This relation can be obtained in two different ways.
We can derive Eq.~\eqref{eq:2-2} according to $v$ at $v=0$.
Alternatively, we can go back to the system dynamics:
queue $2$ is empty at the end of some time slot $t \ge 1$ if
it does not receive any external arrival during this time slot
and, at the end of time slot $t-1$,
queue $1$ was empty and queue $2$ contained at most one packet.

This equation gives the relation
\begin{align*}
  \Phi(u,0,z) &=  1 + \frac{A(u)}{A(0)}[\Phi(0,0,z)- 1].
\end{align*}
We are now in a position to apply the kernel method,
by defining first $U(v,z) = zA(U(v,z))B(v)$ and then $V(z) = z A(V(z)) B(V(z))$.
The generating function $\Pi$ of the stationary distribution of the number of packets in queue 2 can be computed similarly to the previous cases.
For simplicity, we only give its asymptotic behavior:

\begin{theorem}
  With $\delta$ previously defined, for $\Y\sim\Pi$, 
  \begin{equation}
    \label{eq:double_server}
    \proba(\Y\geq R) \underset{R\to\infty}{\sim} \frac{ (1-\lambda_A-\lambda_B)\delta(\delta-1)}{(1-B(\delta))  (1-(T_A\circ B)'(\delta))} \delta^{-R} \text{.}
  \end{equation}
\end{theorem}

\subsection{Non-independent arrivals}
\label{subsec:non-indep}

The analysis can be extended to networks with more generic arrival processes.
We take the network of \figurename~\ref{fig:2flows} as an example.
\begin{itemize}
	\item The arrivals of flows $1$ and $2$ may be dependent.
    The global arrival process is then described by a generating function $A(u,v)$ that cannot be written as a product $A(u)B(v)$.
  \item Within each flow, the numbers of arrivals at different time slots may not be i.i.d. anymore.
    Instead, they may be described by a modulated process (which includes modulated Markov On-Off processes).
    The modulation is described by a finite Markov chain.
    The system dynamics are then described by a system of equations on generating functions (one per state of the Markov chain). 
\end{itemize}

\section{Numerical evaluation}
\label{sec:simu}

\label{sec:simu}
We tested our formulas against simulations in three different scenarios:
the single-server case of \S\ref{sec:oneone},
the multiflow single-server case of \S\ref{subsec:onetwo}
and the tandem network of two single-server queues of \S\ref{subsec:several}. 
The service is deterministic.

Performing simulations consists in computing the stationary distribution of the truncated processes (the number of packets in each queue never exceeds 200)
whose dynamics are described by Eqs.~\eqref{eq:dyn1}, \eqref{eq:dyn2} or~\eqref{eq:dyn2q}.
The approximation of the stationary distribution is obtained by iteratively computing the distribution after $t$ steps for a large enough $t$ (the stopping criterion is when the distance in total variation between the $t$-th and the  $t+1$-th distribution is less than $10^{-12}$).

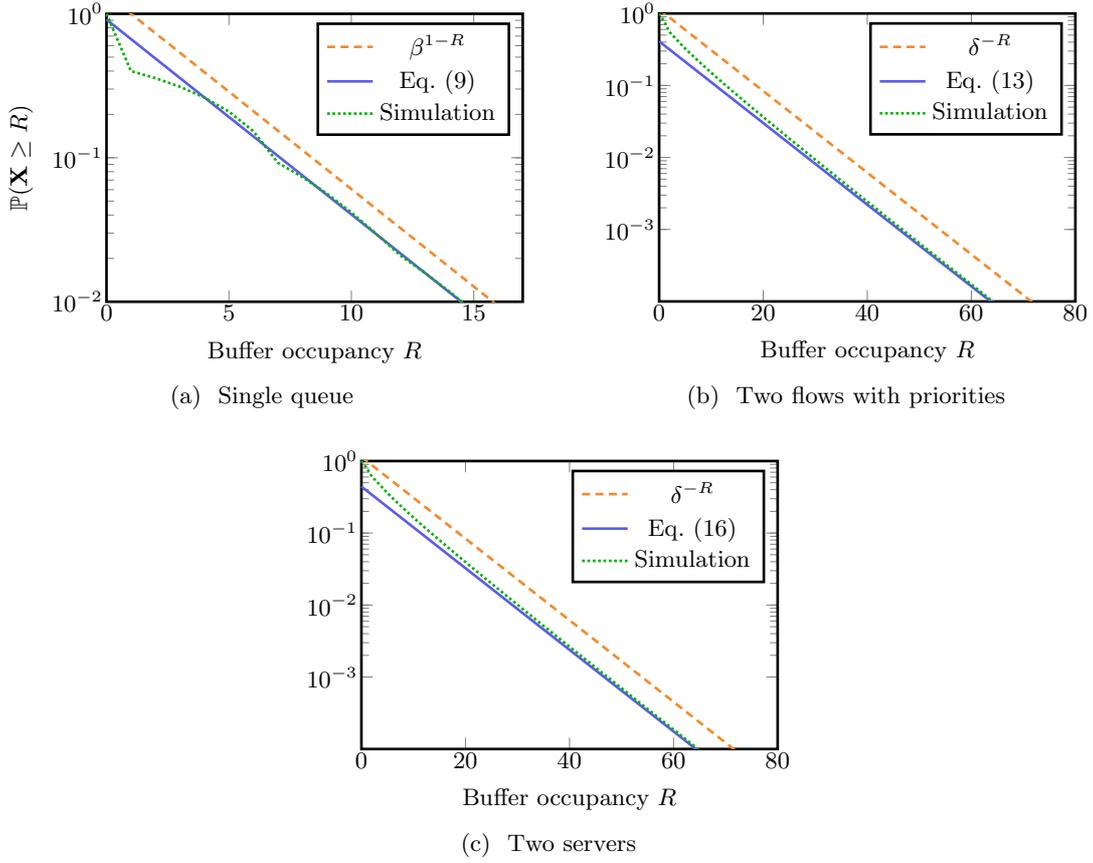
\begin{figure*}[t!]
  \centering
	\subfloat[\label{fig:single} Single queue]{
		\pgfplotstableread[col sep=comma]{file_simple.txt}\r
		\begin{tikzpicture}
		\begin{semilogyaxis}[
		tikzLocality,
		xlabel={Buffer occupancy $R$},
		ylabel={$\proba(\X \geq R)$},
		xmax = 17,
		ymin = .01,
		width=.48\linewidth, height=5.4cm,
		]		
		\addplot[doobs] table [x = R, y = doobs]{\r};
		\addlegendentry{$ \beta^{1-R} $}
		\addplot[noyau] table [x = R, y = noyau]{\r};
		\addlegendentry{Eq.~\eqref{eq:kernel_simple}}
		\addplot[simu] table [x = R, y = simu]{\r};
		\addlegendentry{Simulation}
		\end{semilogyaxis}
		\end{tikzpicture}
	}
  \hfill
	\subfloat[\label{fig:double_file} Two flows with priorities]{
		\pgfplotstableread[col sep=comma]{file_double.txt}\r
		
		\begin{tikzpicture}
		\begin{semilogyaxis}[
		tikzLocality,
		xlabel={Buffer occupancy $R$},
		width=.48\linewidth, height=5.4cm,
    ytick={.001,.01,.1,1.},
		]
		
		\addplot[doobs] table [x = R, y = doobs]{\r};
		\addlegendentry{$ \delta^{-R} $}
		\addplot[noyau] table [x = R, y = noyau]{\r};
		\addlegendentry{Eq.~\eqref{eq:double_file}}
		\addplot[simu] table [x = R, y = simu]{\r};
		\addlegendentry{Simulation}
		
		\end{semilogyaxis}
		\end{tikzpicture}
	}
  \\
	\subfloat[\label{fig:double_server} Two servers]{
		\pgfplotstableread[col sep=comma]{double_server.txt}\r
		\begin{tikzpicture}
		\begin{semilogyaxis}[
		tikzLocality,
		xlabel={Buffer occupancy $R$},
		width=.48\linewidth, height=5.4cm,
    ytick={.001,.01,.1,1.},
		]
		
		\addplot[doobs] table [x = R, y = doobs]{\r};
		\addlegendentry{$ \delta^{-R} $}
		\addplot[noyau] table [x = R, y = noyau]{\r};
		\addlegendentry{Eq.~\eqref{eq:double_server}}
		\addplot[simu] table [x = R, y = simu]{\r};
		\addlegendentry{Simulation}
		\end{semilogyaxis}
		\end{tikzpicture}
	}
	\caption{Numerical evaluation of the kernel method. Parameters $A = D_{2/30, 6}$ and $B = D_{2/5, 1}$.}
	\label{fig:simus}
\end{figure*}

Each arrival process has a bimodal distribution
with generating function $D_{p,M} (u) = (1-p) + pu^M$, for some $p$ and $M$:
at each time slot, either $M$ packets arrive, which occurs with probability $p$, or no packet arrives.
With this distribution, the arrival rate is $D'_{p,M}(1) = Mp$.
In the numerical results of \figurename~\ref{fig:simus}, we take $A = D_{2/30, 6}$ and $B = D_{2/5, 1}$. 
This choice of the functions $ A $ and $ B $ is arbitrary (other distributions lead to similar observations)

Fig.~\ref{fig:single} illustrates the case of a single queue.
The curve $\beta^{1-R}$ is the one that would be obtained using Doob's inequality from~\cite{PC14} (we do not use $\beta^{-R}$ because we consider the size of the queue before service and not after as in~\cite{PC14}). The simulation confirms that we obtain the exact asymptotic behavior, and shows that we improve Doob's inequality by a factor $1.5$.  
We remark that the simulation curve has some irregularities for small values of $R$. This is explained by the arrival of packets in batches of $6$. It is possible to compute the exact error bound from Eq.~\eqref{eq:error} by deriving the first terms explicitly: $[u^R]E(u) = \frac{1}{R!}\frac{d^R}{(du)^R} E(u)\big|_{u=0}$. 

\figurename~\ref{fig:double_file} illustrates the case of a single-server queue with two flows.
We focus on the buffer occupancy of flow 2.
Since flow 1 has priority, its buffer occupancy is still given by Fig.~\ref{fig:single}.
Again, the simulation validates our theoretical results.
Up to our knowledge, there is no formula similar to Doob's inequality, so the curve $\delta^{-R}$ only mimics a \emph{Doob-like inequality}.

Fig.~\ref{fig:double_server} illustrates the error bound in the second queue of the tandem network of \S\ref{subsec:several}.
The error bound differs only from the previous case by a constant factor. 

Although we obtain the exact asymptotic in those two cases, it seems that these are lower bounds of the error.
Indeed, we only computed the first term. But, once again, as we were able to compute an exact formula for the error bounds, more terms are derivable using Taylor expansions.

\section{Conclusion}

In this paper, we have demonstrated on simple examples that methods from analytic combinatorics can be successfully applied to the analysis of queueing systems. 
We have focused on computing backlog bounds, but we believe delay bounds can be derived by using the same techniques as in~\cite{PC14}, for the FIFO, EDF (earliest-deadline-first)  and priorities policies. 
Moreover, combining the computations described in \S\ref{sec:ext} would allow other service policies to enter our framework, in particular some discrete version of GPS (generalized processor sharing).
Following the approach of \cite{BaFl02},
we could also consider a continuous-time extension of our work based on Laplace transforms.
The greatest challenge is to cope with networks of queues. A simple example with two queues has been analyzed. The same analysis can be extended for more than two queues, but this analysis is still partial since the packets are aggregated at each queue. Further investigation needs to be done, in particular in view of the techniques presented in~\cite{BoMi10}.

\end{document}